\documentclass[journal]{IEEEtran}

\usepackage[utf8]{inputenc}
\usepackage[T1]{fontenc}
\usepackage{graphicx}
\usepackage{amsmath, amssymb, amsthm}
\usepackage{algorithm}
\usepackage{algpseudocode}
\usepackage{booktabs}
\usepackage{hyperref}
\usepackage{url}
\usepackage{microtype}
\usepackage{xcolor}
\usepackage{subcaption}
\usepackage{tabularray}
\usepackage{tablefootnote}
\usepackage{eucal}
\usepackage{wrapfig}

\newtheorem{thm}{Theorem}[section]
\newtheorem{lem}[thm]{Lemma}

\newtheorem{defn}{Definition}[section]

\begin{document}

\title{\texttt{Range-Arithmetic}: Verifiable Deep Learning Inference on an Untrusted Party}

\author{
\IEEEauthorblockN{Ali Rahimi$^{\dagger}$ and Babak H. Khalaj$^{\dagger}$}
\IEEEauthorblockA{$^{\dagger}$Department of Electrical Engineering,\\ Sharif University of Technology}
}

\maketitle

\begin{abstract}
Verifiable computing (VC) has gained prominence in decentralized machine learning systems, where resource-intensive tasks like deep neural network (DNN) inference are offloaded to external participants due to blockchain limitations. This creates a need to verify the correctness of outsourced computations without re-execution.
We propose \texttt{Range-Arithmetic}, a novel framework for efficient and verifiable DNN inference that transforms non-arithmetic operations, such as rounding after fixed-point matrix multiplication and ReLU, into arithmetic steps verifiable using sum-check protocols and concatenated range proofs. Our approach avoids the complexity of Boolean encoding, high-degree polynomials, and large lookup tables while remaining compatible with finite-field-based proof systems. Experimental results show that our method not only matches the performance of existing approaches, but also reduces the computational cost of verifying the results, the computational effort required from the untrusted party performing the DNN inference, and the communication overhead between the two sides.
\end{abstract}

\section{Introduction}
By leveraging blockchain technology for transparent coordination, incentives, and auditability, decentralized machine learning (ML) systems enable participants to collaboratively train or infer using models without relying on centralized infrastructure. This shift not only enhances scalability and fault tolerance, but also lays the groundwork for trustworthy and autonomous AI applications across open and dynamic environments. 

A core challenge in decentralized ML systems is ensuring the integrity of computations that are outsourced to external, and potentially untrusted, executors. When decentralized ML systems delegate heavy ML tasks, such as inference or training, to off-chain workers, they must be able to trust the correctness of the results without re-executing the entire computation, which is often prohibitively expensive. This need has led to the application of verifiable computing (VC) techniques in this domain. 

VC techniques enable a \emph{verifier} to confirm the correctness of a computation performed by an untrusted \emph{prover} without re-executing the computation. The three main metrics for evaluating a verifiable computing algorithm are: the size of the communication (or \emph{proof}), the computational cost for the verifier, and the prover's effort in generating the proof. Interactive proofs (IPs) are one of the promising approaches to verifiable computing~\cite{thaler2022proofs, kalai2023snargs, bootle2021sumcheck}. As illustrated in Figure~\ref{fig:overview}, in this approach, the prover sends a claimed result which the verifier challenges over several rounds. The verifier can detect, with high probability, if the claimed result is incorrect, whether due to a discrepancy between the model computed by the prover and the model expected by the verifier, errors during computation, or intentional misrepresentation by the prover. 

VC techniques have gained significant attention from both industry and academia, particularly for secure, trustless, and verifiable outsourced machine learning computation. Notable initiatives include EZKL~\cite{ezkl}, Polygon zkEVM~\cite{polygon}, Accountable Magic~\cite{accountablemagic}, zkAGI~\cite{zkagi}, Noya~\cite{noya}, ZKML Systems~\cite{zkml}, and Provably AI \cite{provably}. Recent academic work has also contributed with surveys and protocols focusing on \emph{verifiable ML}~\cite{peng2025survey, zhan2024validating, lycklama2024artemis, li2024sparsity, weng2023pvcnn}.

\begin{figure*}[t]
\centering
\includegraphics[width=0.8\textwidth]{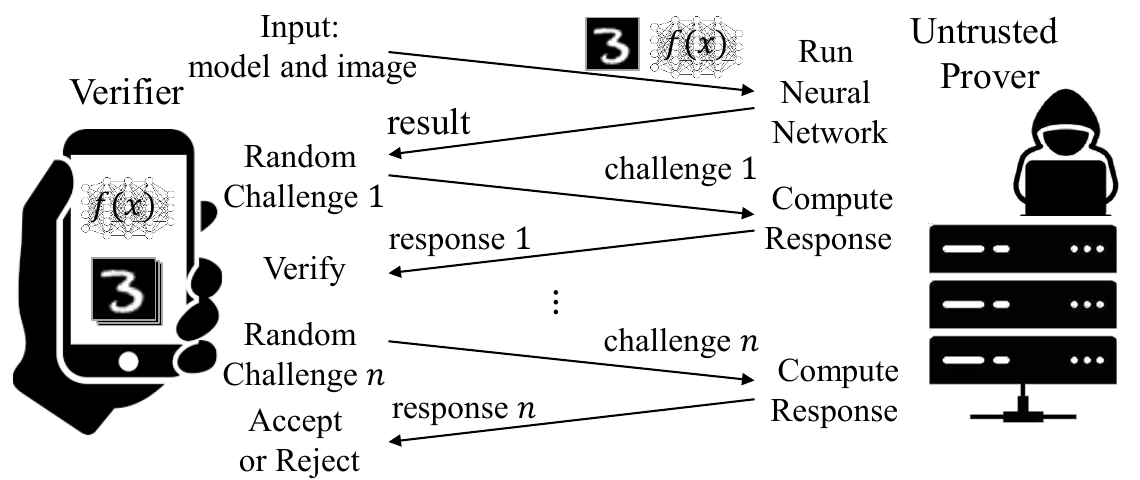}
\caption{A schematic overview of general interactive verifiable computing, where the goal is to compute a function \( f(\cdot) \) at input \( x \) (e.g., performing inference using model \( f(\cdot) \) on input \( x \)). The prover executes the computation and generates a proof of correctness by responding to challenges posed by the verifier.}
\label{fig:overview}
\end{figure*}

One major critique of VC schemes is that they operate over finite fields, often integers modulo a large prime $p$. Furthermore, these schemes are generally restricted to arithmetic computations, that is, computations expressible as compositions of additions and multiplications over the chosen field. However, many widely used operations in ML computations, such as fixed-point arithmetic and non-linear activation functions like ReLU, do not naturally conform to this arithmetic model, making them difficult to represent and verify efficiently within conventional VC frameworks. The following overview of prior work contextualizes our approach within the broader landscape of verifiable ML.

\subsection{Related Work}
\label{sec:related}

To address the challenge of verifying non-arithmetic operations in ML, several lines of work have been explored.

\paragraph{Arithmetic-only approaches.}
Early works such as \cite{weng2021mystique} convert fixed-point operations into binary circuits, incurring $O(n^3 \log n)$ prover overhead and $O(n^2)$ communication due to bitwise decomposition.
The ring-based method of \cite{chen2022interactive} embeds rounding into high-degree polynomials over $\mathbb{Z}_{p^e}$, which, despite $O(\log n)$ communication, introduces a large multiplicative factor $p e$ in prover complexity.
Commit-and-prove schemes \cite{garg2022succinct} and MPC-in-the-head \cite{garg2023experimenting} either expand circuits or simulate multi-party computation, leading to $O(n^2)$ or $O(n^3)$ communication and no ReLU support.
The integer-scaling technique of \cite{dao2024more} (and its precursor \cite{ghodsi2017safetynets}) achieves $O(n^3)$ prover, $O(n^2)$ verifier, and $O(\log n)$ communication, but requires a large finite field to prevent overflow and does not support ReLU.

\paragraph{Specialized ReLU proof systems.}
Recent works have focused specifically on non-linear activations.
\cite{hao2024scalable} introduced lookup-table based proofs for ReLU and sigmoid, reporting over $50\times$ speedups for these functions alone.
However, their method requires large auxiliary tables and does not address fixed-point rounding.
\cite{zkrelu2025} achieves $O(n^2)$ prover complexity by specializing exclusively to ReLU, but it cannot verify matrix multiplication with rounding.
\cite{zip2025} supports high-precision floating-point but uses generic circuit arithmetization, leading to higher constants and $O(n^2)$ communication.
\cite{jolt2026} combines sum-check with lookup arguments, achieving $O(\log n)$ communication, but again relies on precomputed tables for non-linear functions.

\paragraph{General ZKML frameworks.}
Industrial frameworks such as \cite{ezkl} provide end-to-end verifiable inference by compiling models into R1CS constraints.
While they support ReLU and rounding, their generic approach incurs substantial constant overhead – for instance, proof sizes in the kilobyte to megabyte range and prover times orders of magnitude higher than specialized methods (see the benchmark study by \cite{orochi2025}).
\cite{validcnn2025} uses Freivald's technique for matrix multiplication verification, but does not handle rounding or ReLU natively.

\paragraph{Our contribution.}
In this paper, we introduce \texttt{Range-Arithmetic}, a framework that fuses sum-check protocols for arithmetic verification with range proofs for rounding and activation functions. By representing rounding and ReLU through concatenated range constraints, our approach maintains compatibility with finite-field arithmetic and avoids extensive preprocessing, Boolean circuit encodings, or large lookup tables. We show how these operations can be combined in a unified framework to verify ML inference tasks without sacrificing efficiency or generality.

As shown in Table~\ref{tab:comparison}, many recent methods share the same asymptotic complexity classes for matrix multiplication with rounding – $O(n^3)$ prover, $O(n^2)$ verifier, and often $O(\log n)$ or $O(n^2)$ communication. However, \texttt{Range-Arithmetic} is the first to simultaneously achieve:
\begin{itemize}
    \item $O(n^3)$ prover and $O(n^2)$ verifier complexity (matching the best prior work),
    \item $O(\log n)$ communication (matching \cite{dao2024more} and \cite{jolt2026}),
    \item full support for both fixed-point rounding and ReLU without Boolean circuits, large finite fields, or lookup tables,
    \item no preprocessing and no R1CS encoding overhead.
\end{itemize}
While lookup-based methods (\cite{hao2024scalable}, \cite{jolt2026}) may be faster for pure ReLU, they require large tables that are impractical for resource-constrained verifiers. Our method replaces tables with lightweight range proofs, and the square-equality check for ReLU adds minimal overhead. Moreover, unlike \cite{dao2024more}, we do not need to enlarge the field to avoid overflow, making our scheme more practical for cryptographically sized primes.

For numerical comparison, we focus on the state-of-the-art method \cite{dao2024more}, which is the latest in a sequence of progressively improving works and has gained popularity in practice. Since \cite{dao2024more} does not support ReLU, we compare on inference for a linear neural network. In our scheme, rounding is applied after each multiplication, whereas \cite{dao2024more} requires enlarging the finite field. As illustrated in Figure~\ref{fig:mesh8}, \texttt{Range-Arithmetic} achieves superior performance as the number of layers increases. Due to the lack of well-maintained codebases, numerical comparisons with other methods were infeasible.

Thus, \texttt{Range-Arithmetic} provides the most balanced trade-off among prover effort, verifier cost, communication, and feature completeness for verifiable DNN inference on an untrusted party.

\begin{table*}[t]
\centering
\small
\caption{Comparison of verifiable computation methods for verifying the multiplication of two \( n \times n \) matrices with rounding, evaluated in terms of verifier and prover computational complexity, communication complexity, and support for ReLU verification.}
\label{tab:comparison}
\begin{tblr}{
  width = \linewidth,
  colspec = {Q[c,110]Q[c,85]Q[c,85]Q[c,85]Q[c,85]},
  row{1} = {font=\bfseries, bg=gray!20},
  row{2-Z} = {bg=gray!5},
  hlines,
  vline{1-2,6} = {1pt},
}
Method & Prover Compl. & Verifier Compl. & Comm. Compl. & Verifying ReLU \\
\hline\hline
\SetCell[c=5]{c, bg=gray!10, font=\bfseries} Foundational ZK Protocols \\
Mystique~\cite{weng2021mystique} & $O(n^3 \log n)$ & $O(n^2)$ & $O(n^2)$ & Yes (sum-check) \\
Ring-based~\cite{chen2022interactive} & $O(p e n^3)$ & $O(n^2)$ & $O(\log n)$ & No \\
Commit-and-Prove~\cite{garg2022succinct} & $O(n^3 \log n^3)$ & $O(n^{1.5})$ & $O(n^2)$ & No \\
MPC-in-head~\cite{garg2023experimenting} & $O(n^3)$ & $O(n^3)$ & $O(n^3)$ & No \\
\hline
\SetCell[c=5]{c, bg=gray!10, font=\bfseries} Matrix / CNN Verification \\
Integer-scaling~\cite{dao2024more} & $O(n^3)$ & $O(n^2)$ & $O(\log n)$ & No \\
Hao et al.~\cite{hao2024scalable} & $O(n^3)$ & $O(n^2)$ & $O(n^2)$ & Yes (lookup) \\
ZIP~\cite{zip2025} & $O(n^3)$ & $O(n^2)$ & $O(n^2)$ & Yes (Approximation) \\
Jolt Atlas~\cite{jolt2026} & $O(n^3)$ & $O(n^2)$ & $O(\log n)$ & Yes (lookup+teleport) \\
ValidCNN~\cite{validcnn2025} & $O(n^3)$ & $O(n^2)$ & $O(n^2)$ & No \\
EZKL~\cite{ezkl} & $O(n^3)$ & $O(n^2)$ & $O(\text{KB--MB})$ & Yes (circuit) \\
\hline
\SetCell[c=5]{c, bg=gray!10, font=\bfseries} ReLU‑Specialized Methods \\
zkReLU~\cite{zkrelu2025} & $O(n^2)$ & $O(n^2)$ & $O(n^2)$ & Yes (specialized) \\
\textbf{Proposed Method} & $O(n^3)$ & $O(n^2)$ & $O(\log n)$ & Yes (range proof $+$ sum-check) \\
\end{tblr}
\end{table*}

\begin{figure}[t]
\centering
\includegraphics[width=0.9\columnwidth]{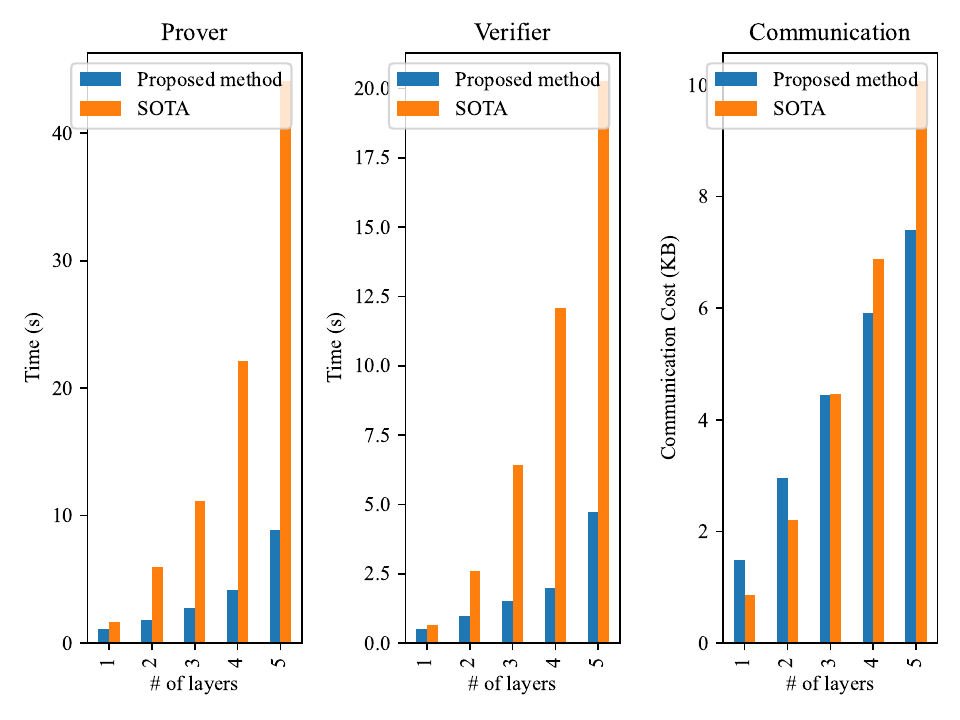}
\caption{Comparison of prover and verifier runtimes, as well as communication costs, for verifiable inference on a linear neural network, relative to the state-of-the-art method in~\cite{dao2024more}.}
\label{fig:mesh8}
\end{figure}

The remainder of the paper is organized as follows. Section~\ref{sec:Preliminaries} reviews the necessary prerequisites and presents an overview of fixed-point arithmetic, including its conversion to finite field representations. We then describe the verification algorithms, covering vector multiplication, polynomial commitment, the Sum-Check protocol, and range-proof techniques. Our proposed algorithm is detailed in Section~\ref{Proposed method}. Section~\ref{Experimentalresults} presents the results of our simulations and evaluates the algorithm’s performance from multiple perspectives. Finally, Section~\ref{Futuredirection} discusses potential directions for future research.

\section{Preliminaries} 
\label{sec:Preliminaries}
\textbf{Notation}. Vectors are denoted by boldface lowercase letters, e.g., $\mathbf{a}$. Matrices are denoted by boldface uppercase letters, e.g., $\mathbf{A}$. Polynomials are denoted by lowercase letters, for example, $a$. We denote the inner product between two vectors by $\left <\mathbf{a}, \mathbf{b} \right >$. For a finite field $\mathbb{F}$, a group $\mathbb{G}$, a  given vector $\mathbf{a} \in \mathbb{F}^n$, and a vector $\mathbf{g} \in \mathbb{G}^n$ of generators, $\mathbf{g}^\mathbf{a} = \prod_{i=1}^n g_i^{a_i}$. For two vectors $\mathbf{a}, \mathbf{b} \in \mathbb{F}^n$, the element-wise product is represented by $\mathbf{a} \circ \mathbf{b}$ which is equal to $\left (a_1 b_1, \dots, a_n b_n \right )$. For $a, b\in \mathbb{R}$, $\left [a,b \right ]$ denotes the set of all numbers $x \in \mathbb{R}$, where  $a \leq x \leq b$. For integers $1 \leq u \leq v \leq n$, and a vector  $\mathbf{a} = (a_1, \dots, a_n)$, then $\mathbf{a}[u-1, v]$ denotes the sub-vector $(a_{u}, \dots ,a_{v})$ of the vector $\mathbf{a}$. Similarly, $a[:u]:= (a_1, \dots,a_u)$, and $\mathbf{a}[u:]:=(a_{u+1}, \dots ,a_n)$.

\begin{defn} 
For any prime number \( p > 2 \), the \textbf{symmetric finite field}, denoted \( \mathbb{F}_p \) (or simply \( \mathbb{F} \)), can be represented by the symmetric complete residue system 
$\left\{ -\frac{p-1}{2}, \ldots, \frac{p-1}{2} \right\}$,
which serves as an alternative to the standard system \( \{0, 1, \ldots, p-1\} \). Field operations such as addition and multiplication are performed modulo \( p \), with the results mapped back into this symmetric representative set.
\end{defn}

\begin{defn}
\label{MLEdef}
Let \( e: \{0,1\}^v \rightarrow \mathbb{F} \) be a function defined over the Boolean hypercube for some integer \( v \). The \textbf{multilinear extension} (MLE) of \( e \), denoted \( \tilde{e}: \mathbb{F}^v \rightarrow \mathbb{F} \), is a multivariate polynomial satisfying the following properties:
\begin{itemize}
    \item For all \( b_1, \dots, b_v \in \{0,1\} \), we have \( \tilde{e}(b_1, \dots, b_v) = e(b_1, \dots, b_v) \).
    \item The polynomial \( \tilde{e}(x_1, \dots, x_v) \) is \emph{multilinear}, i.e., it is linear in each input variable individually.
\end{itemize}
\end{defn}

It has been shown that $\tilde{e}$ can be uniquely written as~\cite{thaler2022proofs}: 
\begin{equation} 
\label{equation_MLE} 
\tilde{e} \left (x_1, \dots, x_v \right ) := \sum_{\mathbf{y} \in \left \{0,1 \right \}^v} e(\mathbf{y}) . \prod_{i=1}^{v} \left (x_i y_i + \left (1 - x_i \right ) \left (1 - y_i \right ) \right ),
\end{equation} 
for any $\left(x_1, \dots, x_v \right )\in \mathbb{F}^v$.

\begin{lem}
\label{Schwartz} 
\textbf{(Schwartz–Zippel Lemma~\cite{schwartz1980fast, zippel1979probabilistic})}  
Let \( f, g: \mathbb{F}^v \rightarrow \mathbb{F} \) be distinct multilinear polynomials. If \( \mathbf{r} \) is sampled uniformly at random from \( \mathbb{F}^v \), then the probability that \( f(\mathbf{r}) = g(\mathbf{r}) \) is at most \( \frac{v}{|\mathbb{F}|} \); that is,
\[
\Pr[f(\mathbf{r}) = g(\mathbf{r})] \leq \frac{v}{|\mathbb{F}|}.
\]
\end{lem}

\subsection{Review of fixed-point arithmetic}
\label{Fixedpointarithmetic}
In this paper, we focus on fixed-point arithmetic, a widely used technique for managing overflow. In this approach, rounding is applied after each multiplication to discard the least significant bits, thereby preventing unbounded growth in memory usage. Specifically, each fixed-point number consists of three components: a sign bit, \( t \) bits for the integer part, and \( s \) bits for the fractional part.

Let \( p \) be a prime with at least \( s + t + 3 \) bits. We represent a real number \( a' \) as a field element \( a = 2^s a' \in \mathbb{F}_p \). When multiplying two fixed-point numbers with \( s \) fractional bits, the result may contain up to \( 2s \) fractional bits. To retain the target precision, a rounding operation truncates the \( s \) least significant bits and rounds to the nearest fixed-point number.

This rounding can be equivalently performed in the field representation using the operator \( \mathfrak{R} \), defined as:
\[
\mathfrak{R}(x) = \frac{x + 2^{s-1} - \left( x + 2^{s-1} \bmod 2^s \right)}{2^s}.
\]
An example of this rounding process is illustrated in Figure 1 in the supplementary material.

\subsection{Inner-product argument}
Consider a system involving a prover and a verifier. Let \( \mathbf{g}, \mathbf{h} \in \mathbb{G}^n \) be two publicly known vectors, consisting of \( 2n \) generators from a group \( \mathbb{G} \). The prover sends two elements to the verifier: \( P \in \mathbb{G} \) and \( c \in \mathbb{F} \), with the goal of convincing the verifier that they possess two vectors \( \mathbf{a}, \mathbf{b} \in \mathbb{F}^n \) such that
\[
P = \mathbf{g}^\mathbf{a} \mathbf{h}^\mathbf{b} \quad \text{and} \quad c = \langle \mathbf{a}, \mathbf{b} \rangle.
\]

To verify this claim, the verifier runs Algorithm 2 (see supplementary material), inspired by the Bulletproofs protocol~\cite{bunz2018bulletproofs}. This algorithm is communication-efficient: the prover transmits only \( 2 \log_2(n) \) group elements, rather than the full vectors \( \mathbf{a} \) and \( \mathbf{b} \), which would require sending \( 2n \) elements. For simplicity, we assume that \( n \) is a power of two.

\subsection{Polynomial commitment}
\label{secPolycommit}
Consider a system comprising a prover and a verifier. The prover holds a multilinear polynomial \( q \), and the verifier wishes to evaluate \( q \) at an arbitrary point \( \mathbf{z} \in \mathbb{F}^v \). A naïve approach would have the prover transmit all coefficients of \( q \), allowing the verifier to compute the evaluation locally. However, this incurs high communication cost and substantial computational effort for the verifier.

A more efficient alternative is to use a polynomial commitment scheme, as described in Algorithm 3 (see supplementary material). In this setting, the evaluation \( q(\mathbf{z}) \) is computed via an inner product. Let \( \mathbf{z} = [z_1, \dots, z_v] \). Then \( q(\mathbf{z}) \) can be written as:
\[
q(\mathbf{z}) = \sum_{\mathbf{i} \in \{0,1\}^v} a_\mathbf{i} \cdot \chi_\mathbf{i}(\mathbf{z}),
\]
where \( a_i \) are the coefficients of \( q \), and \( \chi_i(\mathbf{z}) \) are the Lagrange basis polynomials evaluated at \( \mathbf{z} \). Each basis polynomial is defined as:
\[
\chi_\mathbf{i}(\mathbf{z}) := \prod_{j=1}^{v} \left(z_j i_j + (1 - z_j)(1 - i_j)\right),
\]
for \( \mathbf{i} = (i_1, \dots, i_v) \in \{0,1\}^v \).

The prover commits to the coefficients \( \mathbf{a} = [a_0, \dots, a_{2^v-1}] \) by sending a group element \( P_1 = \mathbf{g}^\mathbf{a} \), where \( \mathbf{g} \) is a vector of public generators. The verifier then checks whether \( q(\mathbf{z}) = \langle \mathbf{a}, \mathbf{b} \rangle \), where \( \mathbf{b} = [\chi_0(\mathbf{z}), \dots, \chi_{2^v-1}(\mathbf{z})] \).

\subsection{The sum-check protocol} 
Suppose the prover and verifier share knowledge of a \( v \)-variate polynomial \( f(x_1, \dots, x_v) \), where each variable has finite degree. The \textbf{sum-check protocol}, used within an interactive proof (IP) system, enables the prover to send a value \( w \in \mathbb{F} \) to the verifier and convince her that
\[
w = \sum_{x_1 \in \{0,1\}} \sum_{x_2 \in \{0,1\}} \cdots \sum_{x_v \in \{0,1\}} f(x_1, \dots, x_v).
\]
This protocol proceeds in \( v \) rounds, with each round consisting of one message from the prover and one from the verifier. The prover’s total computational cost is only a small multiple of the cost of directly computing \( w \), making the protocol highly efficient compared to alternative approaches~\cite{thaler2022proofs}. A detailed description is provided in Algorithm 4 in the supplementary material.

In the final step of Algorithm 4, the verifier must evaluate the multivariate polynomial \( f \) at a randomly chosen point. However, in our setting, downloading all individual coefficients of \( f(\mathbf{x}) \) would incur prohibitive communication overhead. Furthermore, as we later show, the polynomial \( f(\mathbf{x}) \) admits a decomposition:
\[
f(\mathbf{x}) = f^{(1)}(\mathbf{y}_1, \mathbf{x}) \cdot f^{(2)}(\mathbf{x}, \mathbf{y}_2),
\]
where \( \mathbf{y}_1 \) and \( \mathbf{y}_2 \) are fixed inputs, and both \( f^{(1)} \) and \( f^{(2)} \) are multilinear functions.

To mitigate communication cost, we employ binding commitments to the coefficients of \( f^{(1)} \) and \( f^{(2)} \), denoted \( P_1 \) and \( P_2 \), respectively, using the commitment scheme described in the previous subsection. These commitments allow the verifier and prover to independently execute two instances of Algorithm 3 (see supplementary material) to evaluate \( f^{(1)} \) and \( f^{(2)} \) at the desired points. The verifier then computes \( f(\mathbf{x}) \) as the product of the two evaluations.

This leads to Algorithm 5 (see supplementary material), which augments the classical sum-check protocol with polynomial commitments to reduce communication overhead. However, this efficiency gain comes at the cost of increased computational effort for the verifier in the final step. Overall, the total communication complexity of the protocol is dominated by Step 6 of Algorithm 5, and scales with the sum of the degrees of the involved polynomials, i.e.,
$O\left(\sum_{i=1}^{v} \deg(f_i)\right)$
~\cite{bagad2024sum}.

\subsection{Aggregated range proof}
Let \( \mathbf{g} \in \mathbb{G}^{mn} \) be a globally known vector consisting of \( mn \) generators of the group \( \mathbb{G} \). The prover sends an element \( P \in \mathbb{G} \) to the verifier and claims knowledge of a vector \( \mathbf{v} \in \mathbb{F}^m \) such that:
\[
P = \mathbf{g}_{[:m]}^{\mathbf{v}}, \quad \text{where } v_j \in [0, 2^n - 1] \text{ for all } j \in [1, m].
\]
Algorithm 6 in the supplementary material, based on the Bulletproofs protocol~\cite{bunz2018bulletproofs}, enables the prover to efficiently convince the verifier of the validity of this claim.

\subsection{Threat Model}
\label{sec:threatmodel}
We consider a two‑party interactive protocol between an untrusted prover \(\mathcal{P}\) and a probabilistic polynomial‑time verifier \(\mathcal{V}\).  
The prover performs neural network inference on behalf of the verifier and may deviate arbitrarily (e.g., incorrect computation, forged outputs, wrong model).  
The verifier is honest but computationally bounded.  

Public parameters: a finite field \(\mathbb{F}_p\) (prime \(p\) with \(>128\) bits), a group \(\mathbb{G}\) where the discrete logarithm problem is hard, and public generators \(\mathbf{g},\mathbf{h},u\).  
No trusted setup is required beyond random generator selection.

Security requirements:
\begin{itemize}
    \item \textbf{Correctness:} If the prover follows the protocol and computes correctly, the verifier always accepts.
    \item \textbf{Soundness:} If the prover cheats (output is incorrect), the verifier rejects with probability \(1 - \text{negl}(\lambda)\), where \(\lambda\) is the security parameter.
\end{itemize}
Zero‑knowledge is not required – the verifier learns only the input, output, and commitments.

\section{The Proposed \texttt{Range-Arithmetic} Scheme}
\label{Proposed method}
In this section, we introduce the core components of our proposed scheme. We begin by addressing the verification of arithmetic operations, with a focus on matrix multiplication. This is motivated by the fact that many fundamental operations in neural networks, including fully connected layers, convolutions, and batch normalization, can be expressed as matrix multiplications~\cite{bengio2017deep, ju2021efficient}. We then present our approach for handling non-arithmetic operations, specifically rounding, and describe how arithmetic and non-arithmetic components are integrated through a composition step. Next, we extend the framework to support verification of the ReLU activation function. Finally, we outline the complete procedure for verifying inference in neural networks.

\subsection{Verification of Matrix Multiplication Concatenated with Rounding}
\label{sec:problem}
Consider a scenario where the prover holds two matrices, \( \mathbf{A} \in \mathbb{F}^{n \times m} \) and \( \mathbf{B} \in \mathbb{F}^{m \times k} \). The prover computes their product \( \mathbf{C} = \mathbf{A} \mathbf{B} \), and then applies a rounding operation to obtain \( \mathbf{C}^{\prime} = \mathfrak{R}(\mathbf{C}) \). The goal is for the prover to convince the verifier of the correctness of these computations without requiring the verifier to recompute them.

The verification algorithm must satisfy the security properties defined in Section~\ref{sec:threatmodel} (correctness and soundness) and additionally meet the following efficiency criteria:
\begin{itemize}
    \item \textbf{Efficiency:} The communication complexity should be \(O(\log nmk)\), the prover's computational complexity \(O(nmk)\), and the verifier's computational complexity \(O(nm + mk)\).
\end{itemize}

\subsubsection{Verification of Matrix Multiplication}
\label{VerArith}

Inspired by established methods in the literature~\cite{dao2024more, thaler2022proofs}, we reformulate the verification of matrix multiplication $\mathbf{C} = \mathbf{A} \mathbf{B}$ as a sum-check verification problem. Consider three matrices $\mathbf{A} \in \mathbb{F}^{n \times m}$, $\mathbf{B} \in \mathbb{F}^{m \times k}$, and $\mathbf{C} \in \mathbb{F}^{n \times k}$. For simplicity, we assume that $m$, $n$, and $k$ are powers of two. Matrix rows and columns are indexed from $0$ to one less than the total number of rows and columns, respectively.

To reference matrix elements, we define the following functions:
$f_{\mathbf{A}}: \{0,1\}^{\log n \times \log m} \rightarrow \mathbb{F}$, $f_{\mathbf{B}}: \{0,1\}^{\log m \times \log k} \rightarrow \mathbb{F}$, and $f_{\mathbf{C}}: \{0,1\}^{\log n \times \log k} \rightarrow \mathbb{F}$
These functions take the binary representations of row and column indices as input and return the corresponding matrix element. For instance, if $\mathbf{A}$ is a $4 \times 4$ matrix with element $a_{0,2} = 57$, then $f_{\mathbf{A}}([0,0], [1,0]) = 57$, where $[0,0]$ specifies the first row and $[1,0]$ specifies the third column.

Under this formulation, verifying the correctness of the matrix multiplication reduces to checking that, for all $\mathbf{i} \in {0,1}^{\log n}$ and $\mathbf{j} \in {0,1}^{\log k}$, the following equality holds:
\begin{equation}
f_{\mathbf{C}}(\mathbf{i}, \mathbf{j}) = \sum_{\boldsymbol{\ell} \in {0,1}^{\log m}} f_{\mathbf{A}}(\mathbf{i}, \boldsymbol{\ell}) \cdot f_{\mathbf{B}}(\boldsymbol{\ell}, \mathbf{j}).
\end{equation}

As we saw in Definition~\ref{MLEdef}, we can construct MLEs (Multilinear Extensions) $\tilde{a}:\mathbb{F}^{\log n \times \log m} \rightarrow \mathbb{F}$, $\tilde{b}:\mathbb{F}^{\log m \times \log k} \rightarrow \mathbb{F}$, and $\tilde{c}:\mathbb{F}^{\log m \times \log k} \rightarrow \mathbb{F}$ based on functions $f_{\mathbf{A}}$, $f_{\mathbf{B}}$, and $f_{\mathbf{C}}$, respectively. These MLEs are multilinear polynomials that, for binary inputs, yield outputs resembling those of $f_{\mathbf{A}}$, $f_{\mathbf{B}}$, and $f_{\mathbf{C}}$, respectively. Therefore, the above equality is equivalent to establishing the following equality for $\mathbf{i} \in \{0,1\}^{\log n}$ and $\mathbf{j} \in \{0,1\}^{\log k}$
\begin{equation}
\label{MLEeq}
\tilde{c}(\mathbf{i}, \mathbf{j}) = \sum_{\boldsymbol{\ell} \in \{0,1\}^{\log m}} \tilde{a} (\mathbf{i}, \boldsymbol{\ell})  \tilde{b}(\boldsymbol{\ell}, \mathbf{j}).
\end{equation}

We aim to verify \eqref{MLEeq} for all values of $\mathbf{i}$ and $\mathbf{j}$. We observe that since $\tilde{a}$ is linear with respect to the variable $\mathbf{i}$ and $\tilde{b}$ is linear with respect to the variable $\mathbf{j}$, the polynomial on the right-hand side of the equation is linear with respect to variables $\mathbf{i}$ and $\mathbf{j}$. As we discussed, MLE is unique. Therefore, if $\tilde{c}$ is equal to the right-hand side for all possible binary values of $\mathbf{i} \in \{0,1\}^{\log n}$ and $\mathbf{j} \in \{0,1\}^{\log k}$, then \eqref{MLEeq} is of the form of equality of two polynomials. If matrix multiplication $\mathbf{C} = \mathbf{A} \mathbf{B}$ has been performed correctly, then \eqref{MLEeq} must hold for all values in the domain $\mathbf{i} \in \mathbb{F}^{\log n}$ and $\mathbf{j} \in \mathbb{F}^{\log k}$. Thus, we can use the lemma~\ref{Schwartz}, which allows us to verify the equality of two polynomials with high probability by checking their equality at a random point. Assume $\mathbf{r}_1 \in \mathbb{F}^{\log n}$ and $\mathbf{r}_2 \in \mathbb{F}^{\log k}$ are chosen uniformly at random. Then if  
\begin{equation} 
\label{randomMLE}
\tilde{c}(\mathbf{r}_1, \mathbf{r}_2) = \sum_{\boldsymbol{\ell} \in \{0,1\}^{\log m}} \tilde{a} (\mathbf{r}_1, \boldsymbol{\ell}) . \tilde{b}(\boldsymbol{\ell}, \mathbf{r}_2)
\end{equation}
holds, the matrix multiplication $\mathbf{C} = \mathbf{A} \mathbf{B}$ has been computed correctly with the probability of at least $\frac{|\mathbb{F}| - v}{|\mathbb{F}|}$. 

In~\eqref{randomMLE}, the left-hand side is a single number, while the right-hand side is actually a sum over the polynomial $\gamma (\boldsymbol{\ell}) := \tilde{a} (\mathbf{r}_1, \boldsymbol{\ell}) . \tilde{b}(\boldsymbol{\ell}, \mathbf{r}2)$, for all values $\boldsymbol{\ell} \in \{0,1\}^{\log m}$.  Therefore, we can reduce the matrix multiplication verification problem to a sum-check problem over $\sum_{\boldsymbol{\ell}} \gamma (\boldsymbol{\ell})$ and apply Algorithm 5 (see supplementary material) to verify it. It is observed that $\gamma (\boldsymbol{\ell})$ represents the multiplication of two multilinear functions $\tilde{a} (\mathbf{r}_1, \boldsymbol{\ell})$ and $\tilde{b}(\boldsymbol{\ell}, \mathbf{r}_2)$. The prover can compute the commitment to the coefficients of these functions as outlined in Section~\ref{secPolycommit}.

\subsubsection{Verification of Rounding}
Here, we elucidate how the range proof algorithm ensures the validity of the operation $\mathfrak{R}$. It is worth recalling that $\mathfrak{R}(x) = \frac{x + 2^{s-1} - (x + 2^{s-1} \mod 2^s)}{2^s}$, and when $\mathfrak{R}$ is applied to a matrix or vector, the operation is executed on each element individually. We note that $-2^{s-1} \leq (x + 2^{s-1} \mod 2^s) - 2^{s-1} < 2^{s-1}$, which constitutes the truncated portion. In the rounding process, this part is discarded, which makes the numerator of the fraction $\frac{x + 2^{s-1} - (x + 2^{s-1} \mod 2^s)}{2^s}$ a multiple of $2^s$. Subsequently, division by $2^s$ moves its binary representation as $s$ bits to the right. Suppose that matrix $\mathbf{A} \in \mathbb{F}^{\log n \times \log m}$ has been rounded to $\mathbf{A}' = \mathfrak{R}(\mathbf{A})$. Let $\mathbf{E} := \mathbf{A} - 2^s \times \mathbf{A}'$ denote the discarded part during the rounding process, and let $\mathbf{D} := \mathbf{A} - \mathbf{E}$ denote the numerator of the fraction in $\mathfrak{R}$. The correctness of the rounding process is equivalent to ensuring that (1) all entries \( e_i \) of \( \mathbf{E} \) lie within the interval \( -2^{s-1} \leq e_i < 2^{s-1} \), and (2) all elements \( d_i \) of \( \mathbf{D} \) are integer multiples of \( 2^s \).

The first condition limits the discarded number during rounding to a small, standard interval. The second condition ensures $a_i - e_i$ is divisible by $2^s$ without wrap-around. To prevent wrap-around, we require $a'_i = \frac{d_i}{2^s}$ to fall within $-2^{t+1} \leq a'_i < 2^{t+1}$, guaranteeing at least $s$ trailing zeros in $d_i$. This condition suffices to avoid wrap-around, leading to two necessary range proofs, which are checked using an aggregated range proof algorithm. The correctness of this transformation is formally stated in Theorem~\ref{thm:rounding} in Appendix~A.

\subsubsection{Concatenation of Matrix Multiplication and Rounding} 
We now explain how to combine the arithmetic and rounding components using the techniques described above. Suppose that the prover possesses two matrices, $\mathbf{A} \in \mathbb{F}^{n \times m}$ and $\mathbf{B} \in \mathbb{F}^{m \times k}$. After computing their product, $\mathbf{C} = \mathbf{A} \mathbf{B}$, the prover performs rounding operations on the output, resulting in $\mathbf{C'} = \mathfrak{R}(\mathbf{C})$. By executing Algorithm 1 on the input $\mathbf{g}, \mathbf{h} \in \mathbb{G}^{mn}$, $u \in \mathbb{G}$, $\mathbf{A} \in \mathbb{F}^{n \times m}$, and $\mathbf{B} \in \mathbb{F}^{m \times k}$, the prover convinces the verifier of the accuracy of these computations without transmitting the large matrices $\mathbf{A}$, $\mathbf{B}$, $\mathbf{C}$, and $\mathbf{C'}$. In the input tuple, $\mathbf{g}, \mathbf{h} \in \mathbb{G}^{\tau}$ and $u \in \mathbb{G}$ are some globally known generators. Here, $\tau$ denotes the maximum value among $\left (nks, nk(t+1), mk, mn \right )$. This algorithm leverages all the methodologies introduced and examined in prior subsections. 

\begin{figure*}[htbp] 
\centering 
\begin{minipage}{\textwidth} 
\begin{algorithm}[H]
\caption{Verifying a calculation with both arithmetic and non-arithmetic layers}\label{algorithm1}
\begin{algorithmic}[1]
\Require{The algorithm inputs are: $\mathbf{g}, \mathbf{h} \in \mathbb{G}^{\tau}$, $u \in \mathbb{G}$, $\mathbf{A} \in \mathbb{F}^{n \times m}$, $\mathbf{B} \in \mathbb{F}^{m \times k}$.} 
\Require{Verifier has $\mathbf{g}, \mathbf{h} \in \mathbb{G}^{j}$, $u \in \mathbb{G}$.}  

\Require{Prover has $\mathbf{g}, \mathbf{h} \in \mathbb{G}^{j}$, $u \in \mathbb{G}$, $\mathbf{A} \in \mathbb{F}^{n \times m}$, $\mathbf{B} \in \mathbb{F}^{m \times k}$.}  

\Ensure{Verifier receives $P_A$, $P_B$, $P_C$, and $P_{C'}$ as commitments to specific matrices $\mathbf{A}$, $\mathbf{B}$, $\mathbf{C}$, and $\mathbf{C'}$, respectively. Then Verifier accepts that the relations $\mathbf{C} = \mathbf{A} \mathbf{B}$ and $\mathbf{C'} = \mathfrak{R}(\mathbf{C})$ hold between them.} 
\State Prover computes the commitments of $\mathbf{A}$ and $\mathbf{B}$ as $P_A$ and $P_B$ respectively. then sends $P_A$ and $P_B$ to Verifier. 
\State Prover calculates $\mathbf{C} =  \mathbf{A} \mathbf{B}$. 
\State Prover computes the commitment of $\mathbf{C}$ as $P_C$ then sends $P_C$ to Verifier. 
\State Verifier selects two random vectors $\mathbf{r}_1 \in \mathbb{F}^{\log n}$ and $\mathbf{r}_2 \in \mathbb{F}^{\log k}$ and sends $\mathbf{r}_1$ and $\mathbf{r}_2$ to Prover. 
\State Prover constructs the polynomials $\tilde{a}: \mathbb{F}^{\log nm} \rightarrow \mathbb{F}$, $\tilde{b}: \mathbb{F}^{\log mk} \rightarrow \mathbb{F}$, $\tilde{c}: \mathbb{F}^{\log nk} \rightarrow \mathbb{F}$ and $\gamma(\mathbf{z}): \mathbb{F}^{\log m} \rightarrow \mathbb{F}$ according to the description provided in the Section~\ref{VerArith}. 

\Statex \textbf{Note}: The commitments to the coefficients of the polynomial \(\tilde{a}\), \(\tilde{b}\), and \(\tilde{c}\) are represented by $P_A$, $P_B$, and $P_C$ respectively, all of which are already held by the verifier. Furthermore, it is given that $\gamma(\mathbf{z}) = \tilde{a} (\mathbf{r}_1, \mathbf{z}) \tilde{b} ( \mathbf{z}, \mathbf{r}_2)$.

\State Prover and Verifier run Algorithm 5 (see supplementary material) on the input $\left(\mathbf{g}_{[:m]}, \mathbf{h}_{[:m]}, u, P_A, P_B, \gamma \right)$. Verifier obtains $w$.  

\State Prover and Verifier run Algorithm 3 (see supplementary material) on the input $\left (\mathbf{g}_{[:nk]}, \mathbf{h}_{[:nk]}, u, P_C, \tilde{c}, (\mathbf{r_1}, \mathbf{r_2}) \right )$. Verifier obtains $\tilde{c}(\mathbf{r_1}, \mathbf{r_2})$.

\State Verifier checks $w = \tilde{c}(\mathbf{r_1}, \mathbf{r_2})$.

\State Prover computes $\mathbf{C}' = \mathfrak{R}(\mathbf{C})$. 
\State Prover computes the commitment of $\mathbf{C}'$ as $P_{C'}$ then sends $P_{C'}$ to Verifier. 

\State Prover and Verifier run Algorithm 6 (see supplementary material) on the input 
\Statex $\left (
\mathbf{g}_{[:nks]}, \mathbf{h}_{[:nks]}, u, P_C / P_{C'}^{2^s} \times \mathbf{g}_{[:nk]}^{2^{s-1}} \in \mathbb{G}, \mathbf{C} + 2^{s-1} \in \mathbb{F}^{nk}
 \right )$. 
\Statex \Comment{Note that $P_C / P_{C'}^{2^s}$ is the commitment to $E = \mathbf{C} - 2^s \times \mathbf{C}'$. This step verifies for all elements of $\mathbf{E}$ we have $-2^{s-1} \leq e_i < 2^{s-1}$.}

\State Prover and Verifier run Algorithm 6 (see supplementary material) on the input 
\Statex $\left (
\mathbf{g}_{[:nk(t+1)]}, \mathbf{h}_{[:nk(t+1)]}, u, P_{C'} \times \mathbf{g}_{[:nk]}^{2^{t+1}} \in \mathbb{G}, \mathbf{C} + 2^{t+1} \in \mathbb{F}^{nk}
 \right )$. 
\Statex \Comment{This step verifies for all elements of $\mathbf{C}'$ we have $-2^{t+1} \leq c'_i < 2^{t+1}$.}

\Statex If all checks pass, Verifier accepts that $P_A$, $P_B$, $P_C$, and $P_{C'}$ are commitments to certain matrices $\mathbf{A}$, $\mathbf{B}$, $\mathbf{C}$, and $\mathbf{C'}$, respectively, where the relations $\mathbf{C} = \mathbf{A} \mathbf{B}$ and $\mathbf{C'} = \mathfrak{R}(\mathbf{C})$ hold between them.   
\end{algorithmic}
\end{algorithm}
\end{minipage}
\end{figure*}

Our proposed algorithm is adaptable, allowing the prover to reuse it for matrices $\mathbf{A}$ and $\mathbf{B}$ from prior computations or for matrix $\mathbf{C}'$ in subsequent computations without transmitting large intermediary matrices to the verifier. Additionally, the algorithm requires no preprocessing and avoids the need to encode computations using large sparse matrices, unlike techniques such as R1CS or arithmetic circuits~\cite{parno2016pinocchio}.

\subsection{Verification of ReLU activation function} 
\label{VerArith2} 
We verify the ReLU function, defined as $\text{ReLU}(x) = \max\{0, x\} = \frac{x + |x|}{2}$, by applying it element-wise to a matrix. Suppose the verifier has a commitment to a matrix $\mathbf{A} \in \mathbb{F}^{n \times k}$ and receives a commitment to $\mathbf{B} \in \mathbb{F}^{n \times k}$. To verify that $\mathbf{B} = \text{ReLU}(\mathbf{A})$, the prover first computes $\mathbf{Y} = |\mathbf{A}|$ and commits to it. The prover then proves that all elements of $\mathbf{Y}$ are non-negative using a Range proof. Next, the prover proves that for each element $a_{i}$ in $\mathbf{A}$ and $y_{i}$ in $\mathbf{Y}$, $a_{i}^{2} = y_{i}^{2}$ using the following equation:

\begin{equation} 
\label{Equality_relu1} 
\vec{0} = \sum_{\mathbf{s} \in \{0,1\}^{\log nk}} \tilde{I}  \left (\mathbf{x}, \mathbf{s} \right ) . \left ( \tilde{a}^2 \left (\mathbf{s} \right ) - \tilde{y}^2 \left (\mathbf{s} \right ) \right ) \; \; \forall \mathbf{x} \in \{0, 1 \}^{\log nk}
\end{equation}

In \eqref{Equality_relu1}, matrices are encoded using MLE, where $\tilde{I}$ is the MLE of the identity matrix, and $\tilde{a}$ and $\tilde{y}$ encode the elements of $\mathbf{A}$ and $\mathbf{Y}$, respectively. This equation ensures that each element in $\mathbf{Y}$ equals $|\mathbf{A}|$. By the uniqueness of MLEs, the right-hand side polynomial is zero. Instead of checking the equation at every point, the prover and verifier use the sum-check protocol at a random point $\mathbf{r}_1$ chosen by the verifier.

\begin{equation} 
\label{Equality_relu2} 
0 = \sum_{\mathbf{x} \in \{0,1\}^{\log nk}} \tilde{I}  \left (\mathbf{s}, \mathbf{x} \right ) . \left ( \tilde{a}^2 \left (\mathbf{x} \right ) - \tilde{y}^2 \left (\mathbf{x} \right ) \right ) 
\end{equation} 

According to \eqref{Equality_relu2}, the sum-check is performed on the polynomial $f(\mathbf{x}) = \tilde{I}  \left (\mathbf{s}, \mathbf{x} \right ) . \left ( \tilde{a}^2 \left (\mathbf{x} \right ) - \tilde{y}^2 \left (\mathbf{x} \right ) \right )$. The details are provided in Algorithm 7 in the supplementary material. We used Algorithm 8 (see supplementary material) to make ReLU verifiable, which incorporates Algorithm 7 and other components mentioned in this subsection.

\subsection{Verification process of the neural network}
To verify the inference of a neural network, we consider the model as a sequence of layers, each corresponding to an operation, e.g.,  matrix multiplication or ReLU activation. 
Each layer is verified independently. The prover first sends a commitment of the layer’s output to the verifier and then uses this commitment to prove that the corresponding computation was performed correctly, following the proposed algorithm. This modular approach enables the composition of verification algorithms across different layers, allowing them to be reused and combined in arbitrary order and quantity. Please note that output of each layer is the input to the next layer and verifier has access to the inputs. It is important to note that the output of each layer serves as the input to the subsequent layer. An illustrative case study demonstrating this process is presented in Section~\ref{Experimentalresults}.

\section{Security Analysis}
\label{sec:security}

We analyse the security of each building block and their composition.

\subsection{Matrix Multiplication Verification}
The verification reduces to checking a polynomial identity via the Schwartz–Zippel lemma (Lemma~\ref{Schwartz}).  
The verifier chooses random \(\mathbf{r}_1,\mathbf{r}_2\). If the claimed product \(\mathbf{C}\) is incorrect, then  
\[
\Pr[\tilde{c}(\mathbf{r}_1,\mathbf{r}_2) = \sum_{\boldsymbol{\ell}} \tilde{a}(\mathbf{r}_1,\boldsymbol{\ell})\tilde{b}(\boldsymbol{\ell},\mathbf{r}_2)] \le \frac{v}{|\mathbb{F}|}.
\]  
The sum‑check protocol (Algorithm 5 in the supplementary material) is sound: a cheating prover can force acceptance with probability at most \(\frac{\text{poly}(v)}{|\mathbb{F}|}\) per round~\cite{thaler2022proofs}.  
The polynomial commitment and inner‑product argument rely on the discrete logarithm assumption and are sound under the Bulletproofs analysis~\cite{bunz2018bulletproofs}.  
Overall, the probability of falsely accepting a wrong matrix multiplication is bounded by \(\frac{v}{|\mathbb{F}|} + \text{negl}(\lambda)\).

\subsection{Rounding Verification}
Rounding is verified via two aggregated range proofs (Algorithm 6 in the supplementary material) that check:
\[
-2^{s-1} \le e_i < 2^{s-1} \quad\text{and}\quad -2^{t+1} \le c'_i < 2^{t+1}.
\]  
Theorem~\ref{thm:rounding} (see Appendix~A) shows that if both conditions hold then \(c'_i = \mathfrak{R}(c_i)\) with no wrap‑around modulo \(p\).  
The range proof (Bulletproofs) is perfectly complete and computationally sound under the discrete logarithm assumption; its soundness error is negligible in the group order.

\subsection{ReLU Verification}
ReLU is verified by (1) a range proof that all entries of \(\mathbf{Y}=|\mathbf{A}|\) are non‑negative, and (2) a sum‑check proving \(a_i^2 = y_i^2\) via the polynomial  
\(f(\mathbf{x}) = \tilde{I}(\mathbf{s},\mathbf{x})(\tilde{a}^2(\mathbf{x}) - \tilde{y}^2(\mathbf{x}))\).  
If both hold, then \(\mathbf{Y}=|\mathbf{A}|\) and the final check \(P_B^2 = P_A \cdot P_Y\) ensures \(\mathbf{B} = (\mathbf{A}+\mathbf{Y})/2 = \text{ReLU}(\mathbf{A})\).  
Soundness follows from the same arguments as above.

\subsection{Composition}
The protocol verifies each layer independently. A single incorrect layer causes rejection with high probability. For a network of constant depth, the total soundness error remains negligible (by union bound).

\subsection{Security Parameters}
We instantiate with \(|\mathbb{F}_p| > 2^{128}\) (e.g., a 256‑bit prime) so that \(\frac{v}{|\mathbb{F}|} < 2^{-100}\). The group \(\mathbb{G}\) is chosen as an elliptic curve (e.g., secp256k1 or BLS12‑381) with 128‑bit security. Under these parameters, the protocol achieves computational soundness with security level equivalent to 128 bits.

\section{Experimental results} 
\label{Experimentalresults} 
This section evaluates the performance of our proposed algorithm for fixed-point matrix multiplication, in which standard matrix multiplication is followed by a rounding step. The evaluation focuses on three key metrics: the prover’s runtime, the verifier’s runtime, and the communication cost between them. All algorithms are implemented in Python and tested on an Asus X515 laptop. The implementation is available on GitHub\footnote{https://github.com/trainingzk/Range-Arithmetic}.

Figure~\ref{fig:mesh3} presents the performance of Algorithm 1 in verifying the relation \( \mathbf{C'} = \mathfrak{R}(\mathbf{A} \mathbf{B}) \), where \( \mathbf{A} \) and \( \mathbf{B} \) are \( 64 \times 64 \) matrices. As the numeric range increases, the cost of the rounding operation—dominated by the complexity of Algorithm 6 (see supplementary material)—surpasses that of the matrix multiplication.

Figure~\ref{fig:mesh4} illustrates the effect of matrix size on the prover and verifier runtimes. As expected, both increase with matrix size, while the communication overhead grows logarithmically.

Finally, we compare our method to the state-of-the-art ~\cite{dao2024more}, which verifies sequential matrix multiplications of depth \( n \). Unlike method of ~\cite{dao2024more}, which lacks support for rounding and incurs significant storage overhead, our approach integrates rounding efficiently and reduces computational cost for both prover and verifier as the depth increases. This advantage is depicted in Figure~\ref{fig:mesh8}.

\begin{figure}[ht]
\centering
\includegraphics[width=0.9\columnwidth]{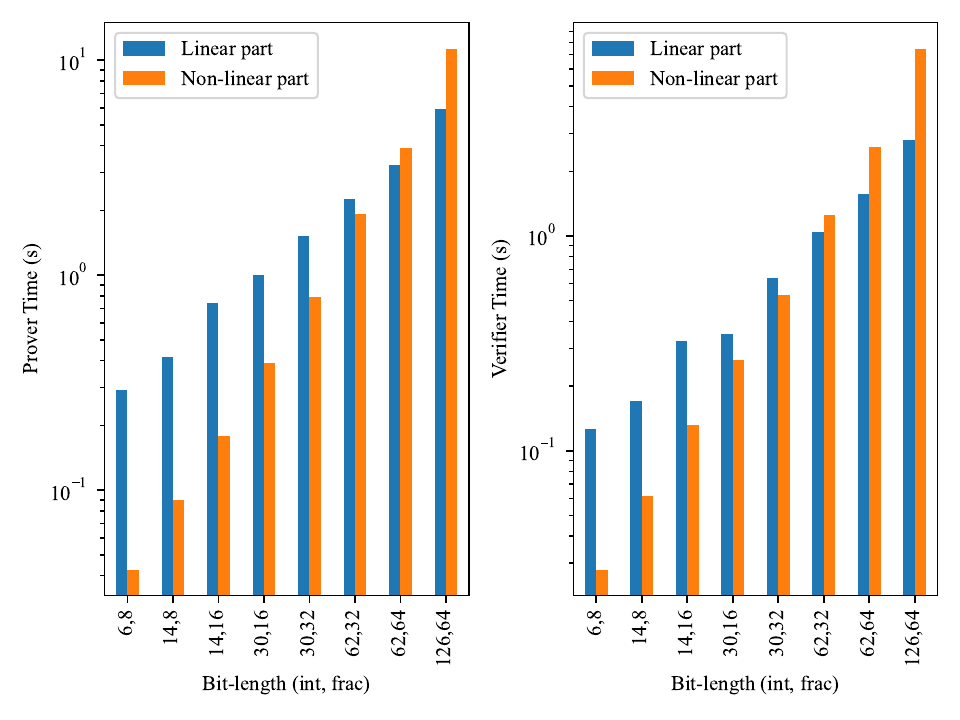}
\caption{Runtime of the arithmetic and non-arithmetic parts for the verifier and the prover in the matrix multiplication.}
\label{fig:mesh3}
\end{figure}

\begin{figure}[ht]
\centering
\includegraphics[width=0.9\columnwidth]{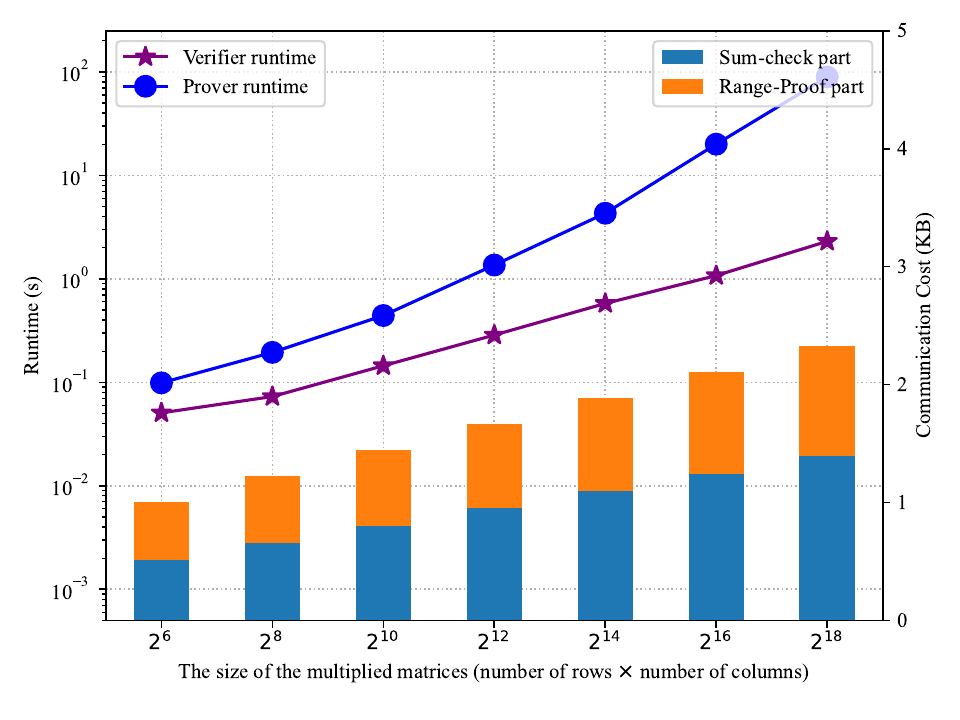}
\caption{Impact of matrix size on the prover and the verifier runtime, as well as communication load.}
\label{fig:mesh4}
\end{figure}

\subsection{A Case Study}
\label{sec:Casestudy}
In the previous section, we evaluated the performance of our proposed method for verifying matrix multiplication. Building on this foundation, we now demonstrate its application to verifying the inference process of a neural network trained on the MNIST dataset~\cite{deng2012mnist}. The network consists of four fully connected layers, each performing matrix multiplication using fixed-point arithmetic, followed by a ReLU activation function. The input vector has a dimension of 784, and the model contains approximately 10{,}000 parameters.

All computations are carried out in fixed-point arithmetic, with weights and inputs represented using 6 integer bits and 8 fractional bits. The trained model achieves an accuracy of approximately 95\%. For verification, matrix multiplications are checked using Algorithm 1, while ReLU activations are verified using a combination of range proofs and the sum-check protocol, as detailed in Algorithm 8 (see supplementary material) and Subsection~\ref{VerArith2}. The prover's runtime is approximately 230 milliseconds, and the verifier's runtime is approximately 154 milliseconds.

\section{Future directions}
\label{Futuredirection}
In this paper, we focused on the verifying inference on  neural network. Future work includes enriching the model to support a wider range of network architectures and layers, incorporating privacy-preserving features, benchmarking against alternative verification schemes, exploring additional evaluation metrics, and extending the approach to other collaborative or trust-sensitive environments such as federated learning.

\bibliographystyle{IEEEtran}
\bibliography{references}

\appendix
\section{Rounding correctness theorem}
\label{app:rounding}

\begin{thm}
\label{thm:rounding}
Let \( a \) be a fixed-point number with 1 sign bit, \( t \) integer bits, and \( s \) fractional bits. Suppose the prover sends integers \( a'_1 \) and \( e_1 \) such that the following conditions hold:
\begin{itemize}
    \item \textbf{Condition 1:} \( a \equiv e_1 + 2^s \cdot a'_1 \pmod{p} \), where \( p \) is a prime with at least \( t + s + 3 \) bits.
    \item \textbf{Condition 2:} \( -2^{t+1} \leq a'_1 < 2^{t+1} \).
    \item \textbf{Condition 3:} \( -2^{s-1} \leq e_1 < 2^{s-1} \).
\end{itemize}
Then the verifier can conclude that \( a'_1 = \mathfrak{R}(a) \).
\end{thm}

\begin{proof}
According to the definition, 
\[
\mathfrak{R}(a) = \frac{a + 2^{s-1} - (a + 2^{s-1} \mod 2^s)}{2^s},
\]
if the prover has performed the calculations correctly, they would obtain the values \( a'_2 = \mathfrak{R}(a) \) and \( e_2 = (a + 2^{s-1} \mod 2^s) - 2^{s-1} \). It is straightforward to verify that these two values satisfy the aforementioned conditions. We will now proceed to demonstrate that \( a'_2 = a'_1 \) and \( e_1 = e_2 \).

From Condition 1, we observe that
\[
a \equiv e_1 + 2^s \cdot a'_1 \equiv e_2 + 2^s \cdot a'_2 \pmod{p},
\]
which implies that 
\[
p \mid (e_2 - e_1) + 2^s \times (a'_2 - a'_1).
\]
Moreover, from Condition 2, it follows that
\[
-2^{t+1} \leq a'_2 \leq 2^{t+1} - 1,
\]
and
\[
-2^{t+1} + 1 \leq -a'_1 \leq 2^{t+1},
\]
which leads to
\[
-2^{t+2} + 1 \leq a'_2 - a'_1 \leq 2^{t+2} - 1.
\]
In addition, from Condition 3, we know that
\[
-2^{s-1} \leq e_2 \leq 2^{s-1} - 1,
\]
and
\[
-2^{s-1} + 1 \leq -e_1 \leq 2^{s-1},
\]
which implies that
\[
-2^{s} + 1 < e_2 - e_1 < 2^s - 1.
\]
Thus, we have the following inequality:
\[
-2^{t+s+2} + 2^s \leq 2^s (a'_2 - a'_1) \leq 2^{t+s+2} - 2^s.
\]
From these inequalities, we deduce that
\[
-2^{t+s+2} + 1 \leq (e_2 - e_1) + 2^s \times (a'_2 - a'_1) \leq 2^{t+s+2} - 1.
\]
Since \( p \) has at least \( s + t + 3 \) bits, we can conclude that
\[
2^{s+t+2} \leq p.
\]
Thus, since \( p \) divides \( (e_2 - e_1) + 2^s \times (a'_2 - a'_1) \), and this expression is smaller in magnitude than \( p \), the only possible solution is
\[
(e_2 - e_1) + 2^s \times (a'_2 - a'_1) = 0.
\]
Given the bounds on \( e_1 \) and \( e_2 \) as well as the multiple of \( 2^s \), it follows that
\[
e_1 = e_2,
\]
and consequently,
\[
a'_1 = a'_2.
\]
\end{proof}

\end{document}